\documentclass{article}
\usepackage{spconf,amsmath,epsfig}
\ninept
\usepackage{amsthm}
\usepackage{amsfonts}
\usepackage{amsmath}
\usepackage{graphicx}
\usepackage{mathrsfs}
\usepackage{color, soul} 
\usepackage{verbatim}
\usepackage{array}
\usepackage{amssymb}

\title{Efficient Recovery of Block Sparse Signals via Zero-point Attracting Projection}
\name{Jingbo Liu, Jian Jin, Yuantao Gu
\sthanks{This work was partially supported by National Natural Science Foundation of China (NSFC 60872087 and NSFC U0835003). The corresponding author of
this paper is Yuantao Gu (gyt@tsinghua.edu.cn).}}

\address{State Key Laboratory on Microwave and Digital Communications\\
 Tsinghua National Laboratory for Information Science and Technology\\
 Department of Electronic Engineering, Tsinghua University, Beijing 100084, CHINA}
\begin{document}

\maketitle

\begin{abstract}
In this paper, we consider compressed sensing (CS) of block-sparse signals, i.e., sparse signals that have nonzero coefficients occurring in clusters.
An efficient algorithm, called zero-point attracting projection (ZAP) algorithm, is extended to the scenario of block CS. The block version of ZAP algorithm employs an approximate $l_{2,0}$ norm as the cost function, and finds its minimum in the solution space via iterations. For block sparse signals, an
analysis of the stability of the local minimums of this cost function under the perturbation of noise reveals an advantage of the proposed algorithm
over its original non-block version in terms of reconstruction error. Finally, numerical experiments show that the proposed algorithm outperforms other state of the art methods for the block sparse problem in various
respects, especially the stability under noise.
\end{abstract}


\section{Introduction}
Compressed sensing (CS) \cite{Donoho}, \cite{Candes} addresses the
problem of retrieving sparse signals from under-determined linear
measurements. It enjoys the advantage of reducing computational
complexity in the measurement stage, and therefore has shown a great
potential in applications such as MRI imaging \cite{imaging}, wireless communication \cite{wireless}, pattern recognition \cite{pattern}, and source coding \cite{sourcecoding}. On the part of signal
reconstruction in CS, one of the key problems is to retrieve the
sparsest solution, i.e., the minimum $l_0$ norm solution to the
equations of linear constraints:
\begin{equation}
\min \|{\bf x}\|_0~~~~\mathrm{s.t.}~~{\bf y}={\bf Ax}, \label{CS}
\end{equation}
where ${\bf x}\in {\mathbb R}^n$ is the unknown sparse signal,
$\mathbf{y}\in {\mathbb R}^m$ is the measurement, and typically
$m<n$. Unfortunately, $l_0$ norm minimization problem is generally
an NP hard problem. Previous work 
including \cite{DC} and
\cite{Donoho} have shown that under some conditions, the sparsest
solution can be obtained via convex relaxation, such as Basis
Pursuit (BP). Another popular method for CS recovery problem is
based on greedy pursuits, and its representative is orthogonal matching
pursuit (OMP) \cite{OMP}.

The block spare problem for compressed sensing was first introduced
by Eldar et.al in \cite{Eldar2}. The authors have shown that
sampling problems over unions of subspaces can be converted into
block-sparse recovery problems. Examples in applications can be
found in \cite{26}, \cite{16} and \cite{17}. Mathematically, a
block-sparse signal can be represented as follows:
\begin{equation}\label{blocksignal}
    {\bf x}=\left[{\bf x}_1^T,{\bf x}_2^T,\cdots,{\bf x}_N^T \right]^T,
\end{equation}
where ${\bf x}_i=\left[x(iD-D+1),\cdots,x(iD)\right]$ is the $i$th
block of $\bf x$ with length $D$, and $n=N\times D$. A signal is $K$
block-sparse if at most $K$ out of the $N$ blocks of the signal are
non-zero. Similar to (\ref{CS}), the $l_{2,0}$ norm minimization for
the block-sparse problem can be formulated as:
\begin{equation}\label{bcs}
    \min \|{\bf x}\|_{2,0},~~~~\mathrm{s.t.}~~{\bf y}={\bf Ax},
\end{equation}
where the $l_{p,q}$ norm of a block vector $\bf x$ is defines as:
\begin{equation}
\|{\bf x}\|_{p,q}=\sum_{k=1}^{N}{(\|{\bf x}_k\|_p)^q}.
\label{l20norm}
\end{equation}
From (\ref{l20norm}) it is clear that the $l_{2,0}$ norm can be interpreted as the number of non-zero blocks of the signal. Like the $l_0$ minimization problem (\ref{CS}), solving (\ref{bcs}) is also NP-hard. Although all the
conventional recovery algorithms in CS is also applicable to the
block CS problem, these algorithms cannot take advantage of the
essential block-sparse characteristic of signals. To remedy this,
Eldar et.al introduced two algorithms in \cite{Eldar} and
\cite{Eldar2}: the $l_{2,1}$-opt and the Block orthogonal matching
pursuits (BOMP). However, like their ancestors, these algorithms
have their inherit drawbacks: $l_{2,1}$-opt is quite slow and
becomes worse as dimension increases; BOMP is faster, but its
estimation accuracy is poorer in the presence of noise perturbation.
%

In contrast, a recently
proposed algorithm called zero-point attracting projection (ZAP)
\cite{jin} is an efficient sparse reconstruction method based on an idea different from the aforementioned convex relaxations and greedy pursuits: The authors
choose a smooth function to approximate the $l_0$ norm and then
finds its minimum in the solution space via iterations. Their
simulations show that ZAP requires fewer measurements for exact
reconstruction than the referenced algorithms in the experiment settings, while
having tractable computational complexity. It is then interesting to
explore the block-sparse reconstruction methods based on the idea of
ZAP.

In this paper, the ZAP algorithm is extended to the block sparse model.
The block ZAP algorithm (BZAP) employs a smoother cost function to
approximate the $l_{2,0}$ norm of the block sparse input, and then minimize this function via
iterations. An analysis of the stability of the local minimums of the
cost function gives a lower bound on the reconstruction error for BZAP than the original ZAP, approximately by a factor of $1/\sqrt{D}$.
Simulations show that BZAP
outperforms other state-of-the-art methods for the block sparse
problem (BOMP, $l_{2,1}$-opt) both in terms of incidence of exact
recovery in noiseless case, and the mean square deviation in the
case of noise-contaminated measurements.
%

The remainder of the paper is organized as follows: section 2
presents the formulation of the BZAP algorithm. In section 3, an
analysis of the $l_2$ stability of the local minimum of the cost
function for BZAP is offered. Section 4 presents simulation results
comparing BZAP with BOMP, $l_{2,1}$-opt and the original ZAP.
Finally, the whole paper is concluded in section 5.

\textbf{Notation.} Throughout this paper, we denote vectors by
boldface lowercase letters, and matrices by boldface upper case
letters. Given a matrix $\mathbf{A}$, $\mathbf{A}^*$ is its
Hermitian conjugate. $\mathbf{A}^{\dagger}$ denotes the pseudo inverse
of $\mathbf{A}$, that is, if $\bf A$ has full row rank or full
column rank, then
\begin{equation}
    \mathbf{A}^{\dagger}=\left\{
    \begin{array}{cl} {\bf A}^*({\bf A}{\bf A}^*)^{-1}, & {\bf A}~\text{has full row rank}; \\
    ({\bf A}^*{\bf A})^{-1}{\bf A}^*, & {\bf A}~\text{has full column rank}. \end{array}\right . \label{Falpha}
\end{equation}
Block support $T$ is a subset of $\{1\dots N\}$ indicating the
non-zero blocks of $\mathbf{x}$, and $T^c$ is its complement. We use
$\mathbf{x}_T$ to denote the vector formed by the blocks in
$\mathbf{x}$ indexed by $T$, and $\mathbf{A}_T$ the sub-matrix that
lies in the column blocks indexed by $T$.
Notation $\|\cdot\|$ takes either the Euclid norm of a vector or the $l_2$ operator norm of a matrix.

\section{Block ZAP Algorithm}
This section aims to extend the ZAP algorithm to the block sparse
problem. One chief idea of BZAP is to employ a `smoother' function:
\begin{equation}
    J({\mathbf{x}})=\sum_{k=1}^N F(\|\mathbf{x}_k\|)\label{modifiednorm1}
\end{equation}
to approximate the $l_{2,0}$ norm of $\mathbf{x}$. Of course, there
is a great liberty in the choice of the function $F$ in
(\ref{modifiednorm1}). But to reduce computation complexity, we
select
\begin{equation}
    F_{\alpha}(w)=\left\{
    \begin{array}{cl} 2\alpha |w|-{\alpha^2}w^2 & |w|\leq\frac{\textstyle1}{\textstyle\alpha}; \\
    1 & {\rm elsewhere}, \end{array}\right . \label{Falpha}
\end{equation}
in the implementations, since its derivative is linear. Now from
(\ref{l20norm}) we see that the $l_{2,0}$ norm of $\mathbf{x}$ can
be approximated as:
\begin{equation}\label{l20}
    \|\mathbf{x}\|_{2,0}\approx \sum_{k=1}^N F_{\alpha}(\|{\bf x}_k\|),
\end{equation}
So the problem (\ref{bcs}) is transferred to
\begin{equation}
    \min \sum_{k=1}^N F_{\alpha}(\|{\bf x}_k\|),~~{\rm s.t.}~~
    {\bf y}={\bf A}{\bf x}.\label{problem1}
\end{equation}
Traditional methods of steepest descent together with a `projection'
step can be used to solve (\ref{problem1}). That is, in the $t$th
iteration, the solution is updated along the negative gradient
direction of the sparse penalty, which in effect attracts the solution to the zero point:
\begin{equation}
    \tilde{\bf x}(t+1)={\bf x}(t)-\kappa\cdot\nabla{J}({\bf x}(t)).
    \label{12}
\end{equation}
Since $\tilde{\bf x}(t+1)$ is generally not in the solution space, the
next step is to project it back to the hyperplane of $\bf Ax=y$:
\begin{equation}
    {\bf x}(t+1)={\bf P}\tilde{\bf x}(t+1)+{\bf Q},
    \label{13}
\end{equation}
where $\bf P=I-A^{\dagger}A$ is named as projection matrix and $\bf
Q=A^{\dagger}y$. The attraction step (\ref{12}) and projection step
(\ref{13}) are used alternately in the iterations, hence the name of
zero-point attracting projection. The procedure of BZAP is
summarized in TABLE \ref{Tab1}.

\begin{table}[ht]
\begin{center}
\caption{Procedure Outline of BZAP}\label{Tab1}
\begin{small}
\begin{tabular}{l}
\hline
{\bf Input:} $\alpha,\kappa,{\bf A},{\bf y};$\\
{\bf Initialize BZAP:} ${\bf x}_0(0)={\bf {A^\dagger}y}, t = 0$.\\
{\bf Repeat:} (for time instant $t$); \\
~~~~Update $\tilde{\bf x}(t+1)$ with the zero attraction by (\ref{12}) and (\ref{l20});\\
~~~~Project ${\bf x}(t+1)$ back to the solution space by (\ref{13});\\
~~~~Update the index: $t=t+1$.\\
{\bf Until:} Block ZAP stop criterion is satisfied.\\
 \hline
\end{tabular}
\end{small}
\end{center}
\end{table}

Finally, we remark on the choice of parameters for the BZAP algorithm:

\emph{The choice of $\alpha$}: According to (\ref{Falpha}), parameter $\alpha$
determines the range of effect of the cost function. There is a
tradeoff in the choice of $\alpha$ since a small $\alpha$ leads to a
bad approximation of the $l_{2,0}$ norm, and produces many local
minimums, while an overly large $\alpha$ limits the effective range.
Empirically we have found that BZAP performs the best when
$1/\alpha$ is around the square root of the variation of the
non-zero entries in $\mathbf{x}$.

\emph{The choice of $\kappa$}: The step length $\kappa$ determines the speed of
convergence and the accuracy of the estimation. A large $\kappa$
will result in a fast convergence but a poor estimation. In our
simulations, $\kappa$ is decreased as the iterations approaches
convergence, in order to ensure both speed of convergence and accuracy. More
specifically, we let $\kappa$ decrease by a factor of $\eta$
($\eta<1$) whenever the cost function (\ref{modifiednorm1}) starts
to increase.

\emph{Stop conditions}: The iteration (\ref{12}) and
(\ref{13}) is terminated when any of the two following conditions is
satisfied: (a) The total number of reductions of step length
$\kappa$ reaches a predefined number $C_1$, or (b) The total number
of iterations exceeds a predefined number $C_2$.

\section{Stability of the Local Minimum Point}
In this section, we consider the problem of estimating
$\bar{\mathbf{x}}$ from the following noisy measurements:
\begin{equation}
\mathbf{y}=\mathbf{A}\bar{\mathbf{x}}+\mathbf{v}. \label{11}
\end{equation}
While in \cite{Wang}, the authors have discussed the
convergence of the ZAP iterations, in this work we mainly consider the stability of
local minimums of the cost function of BZAP under noise perturbation.

Suppose $\alpha$ satisfies
\begin{equation}
1/\alpha<\|\bar{\mathbf{x}}_k\|,\quad k\in T, \label{3}
\end{equation}
define the closed ball $B(\bar{\mathbf{x}},d)$ as a neighborhood of
$\bar{\mathbf{x}}$, where
\[
d=\min_{k\in T}(1/\alpha, \|\bar{\mathbf{x}}_k\|-1/\alpha).
\]
Let $L$ be the solution space:
\[
L:=\{{\bf x}\in {\mathbb R}^n:
\mathbf{y}=\mathbf{A}\mathbf{x}\},
\]
where $\mathbf{y}$ is the measurement given in (\ref{11}). Then,
regarding the stability of the local minimizer of (\ref{problem1}) in the
noise-contaminated measurements, we have the following theorem:
\newtheorem{theorem}{Theorem}
\begin{theorem}\label{Theorem1}
Suppose $\bar{\bf x}$ is a block sparse signal and $\mathbf{A}_T$ has full column rank, then the
minimizer $\hat{\mathbf{x}}$ of function (\ref{modifiednorm1}) in the region
$L\bigcap B(\bar{\mathbf{x}},d)$ satisfies
\begin{eqnarray}
\|\hat{\mathbf{x}}-\bar{\mathbf{x}}\|
\le2\sqrt{N}(1+\|\mathbf{A}_T^\dagger\mathbf{A}_{T^c}\|)
\|\mathbf{A}^\dagger\mathbf{v}\| +\|\mathbf{A}_T^\dagger\mathbf{v}\|.
 \label{2}
\end{eqnarray}
\end{theorem}

\begin{proof}
Let
\begin{equation}
    \delta\mathbf{x}=\hat{\mathbf{x}}-\bar{\mathbf{x}}
\end{equation}
be the difference between the local minimum of the cost function and
the real signal.
The aim is to bound $\|\delta\mathbf{x}\|$ with $\mathbf{v}$.
Obviously,
\begin{eqnarray}
\|\delta\mathbf{x}\| \le
\|\delta\mathbf{x}_T\|+\|\delta\mathbf{x}_{T^c}\|. \label{5}
\end{eqnarray}
Then we will derive bounds on
$\|\delta\mathbf{x}_T\|$ and $\|\delta\mathbf{x}_{T^c}\|$ respectively:

First, consider the bound on $\|\delta\mathbf{x}_{T^c}\|$. If
$\mathbf{x}\in B(\bar{\mathbf{x}},d)$, then
\[
\|\mathbf{x}_k-\bar{\mathbf{x}}_k\|<d, \quad k=1\dots N,
\]
and it follows from the definition of $d$ that
\begin{eqnarray}
\|\mathbf{x}_k\|>1/\alpha,&&\quad k\in T; \nonumber\\
\|\mathbf{x}_k\|<1/\alpha,&&\quad k\in T^c.
\end{eqnarray}
Therefore, with the cost functions defined in (\ref{modifiednorm1}),
(\ref{Falpha}), we have
\begin{eqnarray}
\alpha\|\mathbf{x}_{T^c}\|_{2,1} \le \sum_{k\in
T^c}F_{\alpha}(\|\mathbf{x}_k\|) \le 2\alpha
\|\mathbf{x}_{T^c}\|_{2,1}. \label{8}
\end{eqnarray}

Next, we will prove
\begin{equation}
\|\delta\mathbf{x}_{T^c}\|\le
2\sqrt{N}\|\mathbf{A}^\dagger\mathbf{v}\| \label{9}
\end{equation}
by differentiating between the following two situations:

1) If $\|\mathbf{A}^\dagger\mathbf{v}\|\ge d$, then (\ref{9})
automatically holds.

2) If $\|\mathbf{A}^\dagger\mathbf{v}\|<d$, we have
\begin{align}
\|\delta\mathbf{x}_{T^c}\|
&\le\|\delta\mathbf{x}_{T^c}\|_{2,1}\\
&\le\frac{1}{\alpha}\sum_{k\in T^c}F_{\alpha}(\|\delta\mathbf{x}_k\|)\\
&=\min_{\mathbf{x}}\frac{1}{\alpha}\sum_{k\in
T^c}F_{\alpha}(\|\mathbf{x}_k\|),
\textrm{s.t. }\|\mathbf{x}\|<d, \mathbf{A}\mathbf{x}=\mathbf{v}  \label{definition}\\
&\leq\min_{\mathbf{x}}2\|\mathbf{x}\|_{2,1}, \textrm{s.t.} \|\mathbf{x}\|<d, \mathbf{A}\mathbf{x}=\mathbf{v} \label{use8}\\
&\leq2\|\mathbf{A}^\dagger\mathbf{v}\|_{2,1} \label{feasible}\\
&\leq2\sqrt{N}\|\mathbf{A}^\dagger\mathbf{v}\|, \label{1}
\end{align}
where the definition of $\delta \mathbf{x}$ is used in the
derivation of (\ref{definition}), relation (\ref{8}) in the
derivation of (\ref{use8}), and the fact that
$\mathbf{A}^\dagger\mathbf{v}$ is a feasible point for the
constraint of (\ref{use8}) in the derivation of (\ref{feasible}).
To conclude, (\ref{9}) holds in both situations.

Finally, we derive a bound on $\|\delta\mathbf{x}_T\|$. Since
\[
\mathbf{A}_T\delta\mathbf{x}_T=\mathbf{v}-\mathbf{A}_{T^c}\delta\mathbf{x}_{T^c},
\]
we have
\[
\delta\mathbf{x}_T=\mathbf{A}_T^\dagger(\mathbf{v}-\mathbf{\mathbf{A}}_{T^c}\delta\mathbf{x}_{T^c}),
\]
therefore by triangular inequality and (\ref{9}),
\begin{align}
\|\delta\mathbf{x}_T\|&\le\|\mathbf{A}_T^\dagger\mathbf{v}\|
+\|\mathbf{A}_T^\dagger\mathbf{A}_{T^c}\delta\mathbf{x}_{T^c}\|\nonumber\\
&\le\|\mathbf{A}_T^\dagger\mathbf{v}\|+2\sqrt{N}\|\mathbf{A}_T^\dagger\mathbf{A}_{T^c}\|\|\mathbf{A}^\dagger\mathbf{v}\|_2.\label{4}
\end{align}
Then, combining (\ref{1}) and (\ref{4}) yields the final result (\ref{2}).
\end{proof}

Now, we remark on the improvement of BZAP over ZAP: Since ZAP can be seen as the $D=1$ special case of BZAP, when theorem 1 is applied to ZAP, the bound becomes
\begin{eqnarray}
\|\hat{\mathbf{x}}-\bar{\mathbf{x}}\|
\le2\sqrt{n}(1+\|\mathbf{A}_T^\dagger\mathbf{A}_{T^c}\|)
\|\mathbf{A}^\dagger\mathbf{v}\| +\|\mathbf{A}_T^\dagger\mathbf{v}\|.
 \label{2forZAP}
\end{eqnarray}
It will be shown later that the term $\|\mathbf{A}_T^\dagger\mathbf{v}\|$ in (\ref{2}) is equivalent with the error of a so-called `oracle estimator', which gives a lower bound on the mean square error for the recovery problem. The other term, $2\sqrt{n}(1+\|\mathbf{A}_T^\dagger\mathbf{A}_{T^c}\|)
\|\mathbf{A}^\dagger\mathbf{v}\|$ in (\ref{2forZAP}) is reduced by BZAP by a factor of $\sqrt{D}$ in (\ref{2}). This reveals that the reconstruction via BZAP is more stable than via ZAP in the case of noisy measurements.

\section{Simulation Results}
In this section, the proposed BZAP algorithm is compared with the
conventional ZAP, BOMP, and $l_{2,1}$-opt algorithm. In all
examples, the measurement matrix $\textbf{A}$ has $40$ rows and
$100$ columns, with independent entries following the distribution
of $\mathcal{N}(0,1)$. The block is of the size $D=4$. The locations of nonzero blocks in the
unknown sparse signal $\bar{\mathbf{x}}$ are randomly chosen, and
the values of nonzero elements are independently drawn from the
Rademacher distribution.

For the BZAP algorithm, we set $\kappa = 1$, $\alpha = 1$, $\eta =
0.1$, $C_1=4$, and $C_2=1200$; Therefore it's easily checked that condition (\ref{3}) in theorem
1 is always satisfied. For ZAP and $l_{2,1}$-opt, we adopt the same
stop conditions and control of step size as in the implementation of BZAP.

\subsection{Recovery rate for different block sparsity}
In this first experiment, the exact recovery rate for different
algorithms in the noiseless case is compared. We define exact
recovery when the squared deviation$\|{\bar{\mathbf{x}}}-\hat{\mathbf{x}}\|^2/\|{\bar{\mathbf{x}}}\|^2$
is smaller than $10^{-6}$. One thousand independent simulations are
conducted to calculate the empirical exact
reconstruction rate.

As is shown in Fig.1, the proposed BZAP algorithm outperforms all
the other referenced algorithms in the experiment condition. That is, BZAP can achieve exact reconstruction of sparse
signals when there are more non-zero elements: while
other algorithms exactly reconstruct the signal when $K$ is no more
than 3, BZAP can achieve this when $K$ is up to 4. ZAP gives a poor estimation because it is the only one of the algorithms that doest not employ the block sparse nature of the signal.

\subsection{Mean square deviation (MSD) in the presence of noise}
In this experiment, the noise-contaminated measurements is
formulated as in (\ref{11}). The observational signal-to-noise ratio
(SNR) is defined as
\begin{equation}
    {\rm SNR}=10{\lg}\left(\frac{\|\mathbf{A}\bar{\mathbf{x}}\|^2}{\|{
    {\bf v}}\|^2}\right).  \label{SNR}
\end{equation}
In the simulation the SNR ranges from $10$dB to $50$dB. The noise
vector $\mathbf{v}$ is first generated with independent entries
following the normal distribution and then re-scaled to the fit the
designed SNR.

To compare the reconstruction error, the mean-square deviation
associated with different algorithms is calculated, which is defined
as follows:
\begin{equation}
    {\rm MSD}=\frac{\textbf{E}\|\hat{\textbf{x}}-\bar{\textbf{x}}\|^2}{\textbf{E}\|\bar{\textbf{x}}\|^2}.
    \label{MSD}
\end{equation}
To calculate the empirical expectation in (\ref{MSD}), we take the average of the squared norms over $10^5$ independent simulations.

Regarding the MSD lower bound, consider the following oracle
estimator: suppose the support $T$ is known, then the minimum
variance unbiased estimate of $\mathbf{x}$ is the least square
estimate:
\begin{equation}
\hat{\mathbf{x}}=\mathbf{A}_T^{\dagger}\mathbf{y}.
\end{equation}
The reconstruction error is $\|\mathbf{A}_T^\dagger\mathbf{v}\|$, therefore the MSE is given by
\begin{equation}
\mathbf{E}(\|\hat{\mathbf{x}}-\bar{\mathbf{x}}\|^2)=\sigma^2\mathrm{tr}[(\mathbf{A}_T^*\mathbf{A}_T)^{-1}],
\end{equation}
which should be lower than the achievable MSE for any practical estimators.

The simulation results are shown in Fig.2. In this
experiment, the proposed BZAP algorithms again outperforms other
estimators in terms of MSD, and in fact closely follows the oracle
bound. The BOMP algorithm, although guarantees higher exact recovery
rate in the noiseless case, is very unstable under the
perturbation of noise.

\section{Conclusion}
In this paper, we have extended the ZAP algorithm to the
block-sparse problem, by introducing a cost function to approximate
the $l_{2,0}$ norm of the signal. The stability of the local minimum
of the cost function in BZAP is studied, which reveals an advantage
of BZAP over the original ZAP by employing block sparsity of block sparse signals.
Finally, simulation
results show that BZAP out-performs BOMP, $l_{2,1}$-opt and
the original ZAP both in terms of the incidence of exact recovery in the
noiseless case, and the mean square deviation in the noisy measurements.


{\begin{figure}[!htb]
\begin{center}
\includegraphics[width=3.2in]{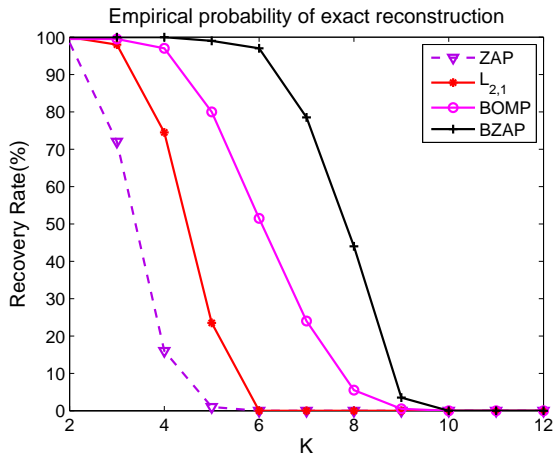}
\caption{Recovery of an input signal from $\bf y=Ax$, where $\bf x$ is a block sparse signal with a block sparsity level of $K$.}
\end{center}
\end{figure}}

{\begin{figure}[!htb]
\begin{center}
\includegraphics[width=3.2in]{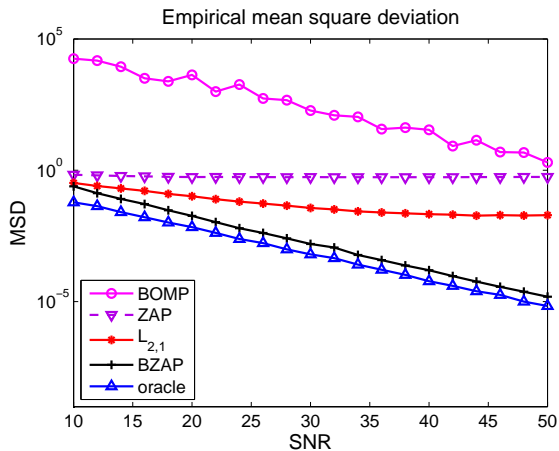}
\caption{Recovery of an input signal from $\bf y=Ax+v$, where $\bf x$ is a block sparse signal with a block sparsity level of $K=4$.}
\end{center}
\end{figure}}


\begin{thebibliography}{1}
\bibitem{Donoho}
D.~L.~Donoho,
``Compressed sensing,''
\emph{IEEE Trans. Inf. Theory},
vol. 52, no. 4, pp. 1289-1306,
April 2006.

\bibitem{Candes}
E.~J.~Cand\`{e}s, J.~Romberg, and T.~Tao,
``Robust uncertainty principles: exact signal reconstruction from highly incomplete frequency information,''
\emph{IEEE Trans. Inf. Theory},
vol.~52, no.~2, pp.~489-509,
Feburary 2006.

\bibitem{imaging}
M.~Lustig, D.~L.~Donoho, and J.~M.~Pauly,
``Sparse MRL: The application of compressed sensing for rapid MR imaging,''
\emph{Magnetic Resonance in Medicine},
vol. 58, no. 6, pp. 1182-1195,
December 2007.

\bibitem{wireless}
W.~U.~Bajwa, J.~Haupt, A.~M.~Sayeed, and R.~Nowak,
``Compressive wireless sensing,''
\emph{Proc. 5th Intl. Conf. on Information Processing in Sensor Networks (IPSN' 06)},
pp134-142, April 2006.

\bibitem{pattern}
J.~Wright, Y.~Ma, J.~Mairal, G.~Sapiro, T.~S.~Huang, and S.~Yan,
¡°Sparse representation for computer vision and pattern recognition,¡±
\emph{Proceedings of the IEEE},
vol. 98, no. 6, pp. 1031-1044,
June 2010.

\bibitem{sourcecoding}
G.~Valenzise, G.~Prandi, M.~Tagliasacchi, and A.~Sarti,
¡°Identification of sparse audio tampering using distributed source coding and compressive sensing techniques,¡± \emph{Journal on Image and Video Processing},
vol. 2009, January 2009.

\bibitem{DC}
E.~J.~Cand\`{e}s and T.~Tao,
``Decoding by linear programming,''
\emph{IEEE Trans. Inf. Theory},
vol. 51, no. 12, pp.4203-4215,
December 2005.

\bibitem{OMP}
J.~A.~Tropp, and A.~C.~Gilbert,
``Signal recovery from random measurements via orthogonal matching pursuit,"
\emph{IEEE Trans. Inf. Theory},
vol. 53, no. 12, pp. 4655 - 4666, December 2007.

\bibitem{Eldar2}
Y.~C.~Eldar, M. Mishali,
``Robust Recovery of Signals From a Structured Union of Subspaces,''
\emph{IEEE Trans. Inf. Theory},
vol. 55, no. 11, pp.5302-5316,
November 2009.


\bibitem{26}
Y.~C.~Eldar, ¡°Compressed sensing of analog signals in shift-invariant spaces,¡±
\emph{IEEE Trans. Signal Processing},
vol. 57, no. 8, pp. 2986 - 2997,
August 2009.

\bibitem{16}
M.~Mishali and Y.~C.~Eldar,
¡°Blind multi-band signal reconstruction: Compressed sensing for analog signals,¡±
\emph{IEEE Trans. Signal Processing},
vol. 57, no. 3, pp. 993¨C1009,
March 2009.

\bibitem{17}
H.~J.~Landau, ¡°Necessary density conditions for sampling and
interpolation of certain entire functions,¡±
\emph{Acta Math}.,
vol. 117, no. 1, pp. 37¨C52, 1967.

\bibitem{Eldar}
Y.~C.~Eldar, P.~Kuppinger, and H.~Bolcskei,
``Compressed sensing of block-sparse signals: Uncertainty relations and efficient recovery,''
\emph{CoRR}, 2009.


\bibitem{jin}
J.~Jin, Y,~Gu, and S.~Mei,
``A stochastic gradient approach on compressive sensing signal reconstruction based on adaptive filtering framework,'' \emph{IEEE Journal of Selected Topics in Signal Processing},
vol.~4, no.~2, pp.~409-420,
April 2010.

\bibitem{Wang}
X.~Wang, Y.~Gu,
``Proof of convergence and performance analysis for zero-point attracting projection algorithm,''
submitted to \emph{IEEE Trans. Signal Processing}.


%
\end{thebibliography}

\end{document}